\documentclass[a4paper,11pt,reqno]{amsart}
\usepackage{enumerate}
\usepackage{graphicx}
\usepackage{amsfonts}
\usepackage{graphicx}
\usepackage{amsmath}
\usepackage{amssymb}
\usepackage{esint}
\usepackage{bbm}
\usepackage[latin1]{inputenc}
\usepackage{hyperref}
\newcommand{\IR}{{\mathbbm{R}}}
\newcommand{\IN}{{\mathbbm{N}}}
\newcommand{\IZ}{{\mathbbm{Z}}}

\newcommand{\FP}{{\mathfrak{P}}}
\newcommand{\Fp}{{\mathfrak{p}}}

\newcommand{\CB}{{\mathcal{B}}}

\newcommand{\CE}{{\mathcal{E}}}
\newcommand{\CC}{{\mathcal{C}}}
\newcommand{\CD}{{\mathcal{D}}}
\newcommand{\CX}{{\mathcal{X}}}
\newcommand{\CW}{{\mathcal{W}}}
\newcommand{\loc}{{\mathrm{loc}}}
\renewcommand{\epsilon}{\varepsilon}
\renewcommand{\phi}{\varphi}
\renewcommand{\rho}{\varrho}
\renewcommand{\theta}{\vartheta}
\renewcommand{\tilde}{\widetilde}

\DeclareMathOperator{\supp}{supp}

\DeclareMathOperator{\diam}{diam}

\newtheorem{Proposition}{Proposition}[section]
\newtheorem{Theorem}[Proposition]{Theorem}
\newtheorem{Definition}[Proposition]{Definition}

\newtheorem{Lemma}[Proposition]{Lemma}
\newtheorem{Corollary}[Proposition]{Corollary}
\newtheorem*{remark}{Remark}

\newcommand{\Hmm}[1]{\leavevmode{\marginpar{\tiny%
$\hbox to 0mm{\hspace*{-0.5mm}$\leftarrow$\hss}%
\vcenter{\vrule depth 0.1mm height 0.1mm width \the\marginparwidth}%
\hbox to 0mm{\hss$\rightarrow$\hspace*{-0.5mm}}$\\\relax\raggedright #1}}}

\begin{document}
\title[]{Heat kernel estimates and related inequalities on metric graphs}
\author[]{Sebastian Haeseler$^1$}
\address{$^1$ Mathematisches Institut, Friedrich Schiller Universit\"at Jena,
  D- 07743 Jena, Germany, sebastian.haeseler@uni-jena.de.}
\begin{abstract}
We consider metric graphs with Kirchhoff boundary conditions. We study the intrinsic metric, volume doubling and a Poincaré inequality. This enables us to prove a parabolic Harnack inequality. The proof involves various techniques from the theory of strongly local Dirichlet forms. Along our way we show Sobolev and Nash type inequalities and related heat kernel estimates.
\end{abstract}
\date{\today} %
\maketitle
\section*{Introduction}
The aim of this article is to study solutions of the heat equation on a metric graph. A metric graph is by definition a combinatorial graph where the edges are considered as intervals, glued together according the combinatorial structure. Metric graphs have received lot of attention in recent years both of the point of view of mathematicians and application (see e.g. the conference proceedings \cite{Proc-05,Proc-08} and the survey \cite{Kuchment-04}).\\
Studying differential operators on metric graphs is similar to studying a system of differential expressions. However the combinatorial structure of the graph is closely related to the spectrum and the heat conductance on the metric graph, see for instance \cite{Cattaneo-95,KostrykinS-06,ExnerPost-05}. Our approach is somewhat different: equipping the (minimal) Laplacian with the so-called standard boundary conditions allows us to treat the metric graph as a strongly local Dirichlet space. We first show the validity of a family of functional inequalities such as Sobolev and Nash type inequalities. The latter one is known to be equivalent to the notion of ultracontractivity of the associated semigroup generated by the Laplacian. This leads to the existence of the heat kernel and an upper Gaussian bound of the heat kernel (see \cite{Davies-89}. This follows the work done in \cite{BoutetdeMonvelLS-07, LenzSchubertStollmann-08} and \cite{Haeseler-08}. \\
The main result of this article will be a parabolic Harnack inequality, which is in fact equivalent to two sided Gaussian estimates of the heat kernel. This also improves the upper Gaussian bound mentioned above in the sense the the first one does not reflect the graph structure whereas the bound coming from the Harnack inequality does. A similar upper bound was already obtained in \cite{Mugnolo-06} in the case of compact graphs, i.e. graphs where the edge set is finite. In \cite{KostrykinPS-07} the heat kernel was explicitly given via a combinatorial series expansion, but the class of graphs is restricted to generalized star graphs, i.e. compact graphs with a finite number of rays attached. Our result will generalize both in the sense that we allow a priori general infinite graphs to obtain asymptotic estimates for the heat kernel. Globally, as we will see, the underlying combinatorial structure will govern the heat conductance, but locally the graph behaves as in one dimension. Note that in \cite{BarlowBass-03} the question is adressed whether the parabolic Harnack inequality is stable for discrete graphs under changing the conductance matrices. Their method uses metric graph techniques, which is called cable system there. As a by-product they show that the parabolic Harnack inequality holds on the metric graph if and only if it holds on the discrete graph. Our approach is different here, since we emphasize on the metric graph while the combinatorial structure will not play a major role. Moreover the setting of Dirichlet spaces allows us to treat arbitrary positive weights on the edges, see subsection \ref{weights}.\\
Our method goes back to the work of Moser (\cite{Moser-61,Moser-63,Moser-71}). These ideas were generalized to the setting of Riemannian manifolds by Saloff-Coste (\cite{SaloffCoste-95,SaloffCoste-92}) and to the setting of strongly local Dirichlet spaces by Sturm (\cite{Sturm-95,Sturm-96}). We use these works to obtain the Harnack inequality. Part of our results goes back to the author's Diploma thesis from 2008. After this work was finished we learned about the recent article \cite{PivarskiSaloffCoste-08} and the preprint \cite{BendikovSaloffCosteSalvatoriWoess-09} which have some overlap with our results.\\
The paper is organized as follows: in section 1 we introduce the notion of metric graphs and the standard Laplacian on it. In section 2 several analytic inequalities are proven, in particular a Sobolev inequality and a Poincaré inequality. The latter one is one half of the necessary input for validity of the Harnack inequality. The first one however helps us also to obtain a Nash inequality and associated heat kernel upper bounds in form of an ultracontractivity estimate of the heat semigroup. In section 3 we will have a further look at the local character of the geometry of the metric graph in terms of volume doubling and Poincaré inequalities on balls. Hence this discussion leads then in section 4 to the statement of the Harnack inequality together with a rough sketch of the proof. Furthermore we state some of its important corollaries, among others we mention two sided Gaussian estimates for the heat kernel. Finally in the last section we have a look at some typical examples.
\section{Basic concepts and results}
\subsection{Metric graphs}
A metric graph $\CX_\Gamma$ consists, like combinatorial graphs, of a countable set $V$ of vertices and of a countable set of edges $E\subset V\times V\setminus\{(v,v)|v\in V\}$. In contrast to the discrete graph we want to consider the edges as intervals, glued together at the vertices. Let $l:E \to (0,\infty)$ the map equipping each edge with a (finite) length, $i:E\to V$ mapping each edge to its initial vertex and $j:E\to V$ mapping each edge to its terminal vertex according to the graph structure. We say that an edge $E$ is incident to a vertex $V$ if either $i(e)=v$ or $j(e)=v$. In order to get a topological structure, we follow \cite{BoutetdeMonvelLS-07} and \cite{Haeseler-08}: let $\CX_e := \{e\}\times(0,l(e))$, $\CX_\Gamma := V\cup \bigcup\limits_{e\in E} \CX_e$ and $\bar{\CX}_e := \CX_e \cup \{i(e),j(e)\}$. Now consider the canonical homeomorphism $\pi_e:\CX_e \to (0,l(e))$, $(e,t)\to t$ and extend it to a homeomorphism $\pi_e:\bar{\CX}_e \to [0,l(e)]$ such that $\pi_e(i(e)) = 0$ and $\pi_e(j(e)) = l(e)$. Note that this forces continuous functions $u\in \CC(\CX_\Gamma)$ to be continuous on each edge interior, i.e. the functions $u_e:= u \circ \pi_e^{-1}$ are continuous and the limits $\lim\limits_{x\to\pi_e(v)} u_e(x)$ are equal for all edges incident to $v$ for all $v\in V$. Furthermore we are able to define a metric structure on the graph. Fix $x,y\in \CX_\Gamma$ and let $\Fp\in \CX_\Gamma^N$ be a polygon connecting $x,y$, i.e. $\Fp_0 = x$, $\Fp_N=y$ and for all $k\in \{1,..,N-1\}$ there exists an edge $e_k$ such that $\Fp_k,\Fp_{k+1}\in \bar{\CX}_{e_k}$. Note that since we do not allow loops, this edge is unique and if we force the graph to be connected there is at least one polygon connecting $x,y$. The length of such a polygon is just the sum of all of its segments, i.e.
\[L(\Fp) := \sum_{k=1}^{N-1} |\pi_{e_k}(\Fp_{k+1}) - \pi_{e_k}(\Fp_k)|\]
and therefore we can define the metric as
\[d(x,y) := \inf \{L(\Fp) : \mbox{ $\Fp$  connects $x$ and $y$}\}.\]
To make sure that this defines a metric we assume from now on that the graph is connected and that the vertex degrees are finite for all vertices, that is
\[d_v := |\{e\in E: v\in \{i(e),j(e)\}\}| < \infty \mbox{ for all } v\in V.\]
In order to get a suitable measure on $\CX_\Gamma$, we define for $Y\subset \CX_\Gamma$
\[\int\limits_Y u(x) \: dm(x) := \sum_{e\in E} \int\limits_{e\cap Y} u(x) dm(x).\]
where $m$ is the the measure induced by the images of the Lebesgue measure on each $(0,l(e))$. The $L^p$ spaces are then given by
\[L^p(\CX_\Gamma,m)= \bigoplus\limits_{e\in E} L^p(0,l(e)).\]
For later purposes, we define the mean of a function $u$ to be
\[ \bar{u}_Y := \fint\limits_Y u(x)\: dm(x) := \frac{1}{m(Y)} \int\limits_Y u(x) \: dm(x).\]
\subsection{The energy form}
Our next aim is to define the energy form. For this purpose we define for $1\leq p \leq \infty$
\[W^{1,p}(\CX_\Gamma) := \bigoplus\limits_{e\in E} W^{1,p}(0,l(e)) \cap \CC(\CX_\Gamma)\]
and
\[W^{1,p}_0 (\CX_\Gamma):=W^{1,p}(\CX_\Gamma)\cap \CC_0(\CX_\Gamma),\]
where $\CC_0(\CX_\Gamma)$ denotes the closure w.r.t. $\|\cdot\|_\infty$ of the set of continuous functions with compact support.
The energy form is then given by
\[\CD = \CD(\CE) = W^{1,2}_0 (\CX_\Gamma),\]
\[ \CE(u,v) := \sum_{e\in E} \int\limits_0^{l(e)} u'_e(x) v'_e(x) \: dx.\]
One can show (see for instance \cite{Kuchment-04}) that this form is associated with the operator acting on functions as $-\frac{d^2}{dx^2}$ on each edge with domain
\[\left\{u\in \CC_0(\CX_\Gamma)\cap \bigoplus_{e\in E} W^{2,2}(0,l(e)) \ :\ \sum_{e\sim v} u'_e(v) = 0\right\},\]
where $u'_e(v)$ denotes the directional derivative towards $v$ along $e$. This operator is known as the Laplacian with Kirchhoff vertex conditions, which we will denote by $\Delta$.\\
Note that we implicitly impose Neumann boundary conditions on vertices with degree one. Furthermore although we explicitly do not allow loops, multiple edges and edges of infinite length, we implicitly do, since we can view any $x\in \CX_\Gamma\setminus V$ as a vertex with degree 2. Functions in the domain of $\CE$ will then automatically fulfill the continuity condition.\\
Note that, $\CE$ is a regular and strongly local Dirichlet form with energy measure
\[d\Gamma(u(x)) = |u'(x)|^2 dm(x) = \sum\limits_{e\in E} |u_e'(x)|^2 \: dx\]
(cf. \cite{BoutetdeMonvelLS-07}). This allows us to define functions locally in the domain, i.e.
\[u\in \CD_\loc :\Leftrightarrow  \ u\in L^2_\loc \mbox{ and } \forall \phi \in \CD\cap C_c(S): \phi u \in \CD.\]
Using the energy measure one can define the intrinsic metric $\rho$ by
\[\rho(x,y) := \sup \{ |u(x) - u(y)| : u\in W^{1,2}_{\loc}(\CX_\Gamma) \cap \CC(\CX_\Gamma) \mbox{ and } |u'| \leq 1 \}.\]
For a discussion of the intrinsic metric in the context of general strongly local Dirichlet forms see \cite{Sturm-94}. In our situation we have the following result:
\begin{Proposition}\label{intrinsic}
For all $x,y\in \CX_\Gamma$ we have
\[\rho(x,y) = d(x,y).\]
\end{Proposition}
\begin{proof} Let $y\in \CX_\Gamma$ and $d_y (x) := d(x,y)$, resp. $\rho_y (x) := \rho(x,y)$. Clearly: $d_y \in \CC(\CX_\Gamma)$ and for $x\in \CX_e$, we have
\[d_y(x) := \min\{ d(y,i(e)) + |\pi_e(x) - \pi_e(i(e))|, d(y,j(e)) + |\pi_e(x) - \pi_e(j(e))|\},\]
which gives $d\Gamma(d_y) \leq dm$ and therefore $d_y \in W^{1,2}_{loc} (\CX_\Gamma)$. Now by definition of $\rho$ we have $d_y(x) \leq \rho_y (x)$ for all $x,y\in \CX_\Gamma$.\\
Let now $x,y \in \overline{\CX}_e$ and $u\in W^{1,2}_{\loc}(\CX_\Gamma) \cap \CC(\CX_\Gamma)$ an arbitrary function with $|u'| \leq 1$. Then we have
\begin{eqnarray*}
|u(x) - u(y)| & =& |u_e (\pi_e(x)) - u_e(\pi_e (y))|\\
&\leq &\int\limits_{\pi_e(x)}^{\pi_e(y)} |u_e'(t)|\: dt\\
&\leq& |\pi_e(x)- \pi_e(y)|,
\end{eqnarray*}
i.e. $\rho(x,y) \leq |\pi_e(x) - \pi_e(y)|$ on $\CX_e$. Choose a path $\Fp$ with $L(\Fp) = d(x,y) + \epsilon$ for some $\epsilon \geq 0$. This yields
\begin{eqnarray*}
\rho(x,y) \leq \sum_{k=1}^{N-1} \rho(\Fp_k,\Fp_{k+1}) \leq \sum_{k=1}^{N-1} |\pi_{e_k}(\Fp_k) - \pi_{e_k} (\Fp_{k+1})| =  d(x,y)+\epsilon,
\end{eqnarray*}
where $e_k$ denotes the unique edge with $\Fp_k, \Fp_{k+1} \subset \CX_{e_k}$. Now the claim follows, since $\epsilon>0$ was chosen arbitrary.
\end{proof}
\subsection{More general models}\label{weights}
We can also treat more general forms as follows: let $c_e:\CX_e\to \IR$ measurable for all $e\in E$. Assume that there exists $\Lambda >0$ such that
\[0<\Lambda^{-1} \leq c_e(x)\leq \Lambda\]
for all $e\in E$ and $x\in \CX_\Gamma$. Then we define the form
\[\CD = \CD(\tilde{\CE}) = W^{1,2}_0 (\CX_\Gamma),\]
\[ \tilde{\CE}(u,v) := \sum_{e\in E} \int\limits_0^{l(e)} u'_e(x) v'_e(x) \: c_e(x)dx.\]
Then, we can compare the energy measure $d\tilde{\Gamma}$ of this form with the energy measure $d\Gamma$ of $\CE$: for $u \in \CD(\CE)=\CD(\tilde{\CE})$ we have
\[\Lambda^{-1} d\Gamma (u)\leq d\tilde{\Gamma}(u) \leq \Lambda d\Gamma(u).\]
This implies for instance the equivalence of the intrinsic metrics of both forms, i.e. for all $x,y\in \CX_\Gamma$ we have
\[\Lambda^{-1} \tilde{\rho}(x,y) \leq d(x,y) \leq \Lambda \tilde{\rho}(x,y).\]
Moreover the Nash, the Poincaré and the parabolic Harnack inequality, see sections 2.2, 3 and 4, together with their consequences hold true with constants depending additionally on $\Lambda$. We will show these results for the form $\CE$ and leave all modifications for the results on $\tilde{\CE}$ to the reader, as they are easily derived from the above mentioned comparison of the energy measures.
\section{One-dimensional inequalities}
\subsection{Sobolev and Poincaré type inequalities}
The next few results reflect the one dimensional behavior in terms of some useful Poincaré and Sobolev-type inequalities.\\
Note that the graph as metric space is compact if and only if the edge set is finite.
\begin{Lemma}
Let $\CX_\Gamma$ be a compact metric graph and $u\in W^{1,1} (\CX_\Gamma)$. Then there exists $c\in \CX_\Gamma$ such that for all $x\in \CX_\Gamma$ and each $\mathfrak{p}_c^x\in \FP$ connecting $c$ and $x$, we have
\[|u(x)-\bar{u}| \leq \int\limits_{\mathfrak{p}_c^x} |u'(y)| \: dm(y)\label{path}\tag{$\clubsuit$}\]
where $\bar{u} = \fint u \:dm$.
\end{Lemma}
\begin{proof}
Since $u$ is continuous it attains its maximum and minimum, i.e. $u(x_+)= \max u$ and $u(x_-) = \min u$. In particular $u(x_-)-\bar{u} \leq 0$ and $u(x_+) - \bar{u} \geq 0$. Without restriction let the inequalities be strict, otherwise we choose the point $c$ below as $x_-$ or $x_+$ resp. Since the graph is connected, there is a path $\mathfrak{p}_{x_-}^{x_+}$ connecting $x_-$ and $x_+$. Note that $u$ restricted to this path is continuous and changes the sign, hence there exists $c\in \mathfrak{p}_{x_-}^{x_+}$ such that $u(c)= \bar{u}$. Take now $x\in \CX_\Gamma$ arbitrary and a path $\mathfrak{p}_x^c$ connecting $c$ and $x$. Then by the fundamental theorem of calculus we have
\[|u(x)- \bar{u}| \leq \int\limits_{\mathfrak{p}_c^x} |u'(y)| \: dm(y).\]
\end{proof}
This result yields a whole family of Poincaré and Sobolev-type inequalities. Recall that the diameter of a metric space is defined as
\[ \diam \CX_\Gamma := \sup \{ d(x,y) : \ x,y\in \CX_\Gamma\}.\]
\begin{Corollary}\label{1.2}
Let $\CX_\Gamma$ be a compact metric graph and $u\in W^{1,p} (\CX_\Gamma)$ for $1\leq p < \infty$. Then
\[\|u-\bar{u}\|_q \leq (\diam \CX_\Gamma)^{1 - \frac{1}{p}} |\CX_\Gamma|^\frac{1}{q} \|u'\|_p \]
for arbitrary $1\leq p\leq \infty$ and $1\leq q\leq \infty$. In particular
\[\|u-\bar{u}\|_p \leq (\diam \CX_\Gamma)^{1 - \frac{1}{p}} |\CX_\Gamma|^\frac{1}{p} \|u'\|_p .\]
Furthermore $u$ fulfills the Sobolev inequality
\[\|u\|_\infty \leq |\CX_\Gamma|^{-\frac{1}{p}} \|u\|_p + (\diam \CX_\Gamma)^\frac{p-1}{p} \|u'\|_p.\]
\end{Corollary}
\begin{proof}
Both families of inequalities follow directly from the previous lemma: By \eqref{path} we get for all $y\in \CX_\Gamma$
\[|u(y)-\bar{u}|  \leq \int\limits_\mathfrak{p} |u'(y)| \: dm(y) \leq |\mathfrak{p}|^{1-\frac{1}{p}} \|u'\|_p \leq |\diam \CX_\Gamma |^{1-\frac{1}{p}} \|u'\|_p.\]
Hence we have proven the first inequality for $q=\infty$. For $1\leq q < \infty$ we then have
\begin{eqnarray*}
\left(\int |u(y) - \bar{u}|^q\: dm(y) \right)^\frac{1}{q} &\leq & |\CX_\Gamma| ^\frac{1}{q} \sup_{y\in \CX_\Gamma} |u(y)-\bar{u}|\\
&\leq& |\diam \CX_\Gamma|^{1- \frac{1}{p}} |\CX_\Gamma| ^\frac{1}{q} \|u'\|_p.
\end{eqnarray*}
This gives the desired inequality in the case $1\leq q < \infty$.\\
Now set $q=\infty$ and use the (inverse) triangle inequality to obtain
\[\|u\|_\infty \leq |\CX_\Gamma|^{-1} \|u\|_1 + |\diam \CX_\Gamma|^\frac{p-1}{p} \|u'\|_p \]
which by the Hölder inequality yields the claim.
\end{proof}
\begin{Corollary}\label{SobEmb}
Let $\CX_\Gamma$ be a compact metric graph. Then the embedding
\[W^{1,p} (\CX_\Gamma) \hookrightarrow \CC_0(\CX_\Gamma)\]
is continuous for $1\leq p <\infty$ and compact for $1< p< \infty$.
\end{Corollary}
\begin{proof}
The continuity of the embedding is clear by the Sobolev inequality in the previous corollary. To show compactness, let $(f_n)$ be a bounded sequence in $W^{1,p}(\CX_\Gamma)$, i.e for all $n\in \IN$ we have $\|f_n\|_{1,p} \leq M$ for some $M>0$. By
\[\|f_n\|_\infty \leq c \|f_n\|_{1,p} \leq cM\]
the sequence is bounded for all $x\in \CX_\Gamma$. Furthermore we have
\[ |f_n(x) - f_n(y)| \leq \int\limits_{\mathfrak{p}_x^y} |f_n'(z)| \: dm(z) \leq L(\mathfrak{p}_x^y)^\frac{1}{q} \biggl( \int\limits_{\mathfrak{p}_x^y} |f'(z)|^p \: dm(z) \biggr)^\frac{1}{p}\]
where $\mathfrak{p}_x^y$ is an arbitrary path connecting $x$ and $y$ and $L(\mathfrak{p}_x^y)$ denotes its length. Hence we get
\[|f_n(x) - f_n(y)| \leq d(x,y)^\frac{1}{q} \|f\|_{1,p}\leq M d(x,y)^\frac{1}{q},\]
i.e. the sequence $(f_n)$ is equicontinuous. The Arzelà-Ascoli theorem now gives the compactness.
\end{proof}
Note that the simplest compact metric graph is a closed interval. The results above are well known for closed intervals. So we have seen that compact metric graphs could be treated like closed intervals. The simplest non-compact metric graph is a half-line. The lemma below shows the analogy of half-lines to non-compact metric graphs.
\begin{Lemma}\label{2.3}
Let $\CX_\Gamma$ be an infinite metric graph and $u\in W^{1,p}(\CX_\Gamma)$ for $1\leq p < \infty$. Then we have
\[\|u\|_\infty \leq (\tfrac{p-1}{p})^\frac{p-1}{p}\|u\|_{1,p}\]
for $1<p<\infty$, and
\[\|u\|_\infty \leq \|u'\|_{1}\]
for $p=1$. Let $Y\subset \CX_\Gamma$ be open and pre-compact, then for all $u \in W^{1,p}_0(Y)$ we have
\[\|u\|_\infty \leq (\tfrac{p-1}{p})^\frac{p-1}{p}\|u\|_{1,p}\]
and
\[\|u\|_\infty \leq \|u'\|_{1}.\]
\end{Lemma}
\begin{proof}
Fix $x\in \CX_\Gamma$. Choose $y_n,z_n\in \CX_\Gamma$ with $d(y_n,z_n)=\mathrm{const.}$ and $d(y_n,x)\to \infty$. Then we have
\[|u(y_n)-u(z_n)| \leq \int\limits_{\mathfrak{p}_{y_n}^{z_n}} |u'(z)| \: dm(z).\]
Since $u\in W^{1,p}(\CX_\Gamma)$ the right hand side tends to zero, and hence $u(y)\to c$ for some constant $c$ for $d(x,y)\to \infty$. But since $u\in L^p(\CX_\Gamma)$ it follows that $c=0$. This implies
\[|u(x)|\leq \|u'\|_1.\]
For $1<p<\infty$ we apply this inequality to $u^p$ and by Hölder inequality we get
\[\| |u|^p \|_\infty^\frac{1}{p} \leq p^\frac{1}{p} \|u^{p-1} u'\|_1^\frac{1}{p} \leq p^\frac{1}{p} \|u\|_p^\frac{p-1}{p} \|u'\|_p^\frac{1}{p} \leq (\tfrac{p-1}{p})^\frac{p-1}{p} (\|u\|_p + \|u'\|_p),\]
where the last estimate follows from the well known inequality $ab \leq \frac{1}{p} a^p + \frac{1}{q} b^q$, for all $a,b >0$. The second statement of the lemma follows form the first by extending $u$ by zero on $\CX_\Gamma$.
\end{proof}
\begin{remark}{\rm As already mentioned the inequalities above reflect the one dimensional structure of the graph. Once \eqref{path} is established, our arguments proceed essentially as in the case of an interval or a semi-axis, see for instance \cite{Burenkov-98}.}
\end{remark}
\subsection{Nash inequality and the heat kernel}
The Sobolev estimates above are the starting point for us to deduce bounds on the heat kernel. The first step will be the following Nash type inequality.
\begin{Lemma}
Let $\CX_\Gamma$ be an infinite metric graph and $u\in W^{1,2}(\CX_\Gamma)\cap L^1(\CX_\Gamma)$. Then there exists $c>0$ such that
\[ \|u\|_2 \leq c \|u'\|_2^\frac{1}{3} \|u\|_1^\frac{2}{3}.\]
Let $Y\subset \CX_\Gamma$ pre-compact, then for all $u \in W^{1,2}(Y)\cap L^1(Y)$ we have
\[ \|u\|_2 \leq c (|\CX_\Gamma|^{-1} \|u\|_2 + \|u'\|_2)^\frac{1}{3} \|u\|_1^\frac{2}{3}.\]
\end{Lemma}
\begin{proof}
Let $x\in \CX_\Gamma$. Then Lemma \ref{2.3} applied to $u^2$ gives
\[|u(x)|^2 \leq 2 \|u u'\|_1 \leq 2 \|u\|_2 \|u'\|_2\]
and hence
\[|u(x)|^2 \leq \sqrt{2} \|u\|_2^\frac{1}{2} \|u'\|_2^\frac{1}{2} |u(x)|\]
which yields by integrating over $\CX_\Gamma$
\[ \|u\|_2^2 \leq \sqrt{2} \|u\|_2^\frac{1}{2} \|u'\|_2^\frac{1}{2} \|u\|_1\]
so the first statement is proven. Analogous calculations starting with the Sobolev inequality
\[\|u\|_\infty \leq |\CX_\Gamma|^{-1} \|u\|_1 + \|u'\|_1.\]
from Corollary \ref{1.2} with $p=1$ yield the second inequality.
\end{proof}
Using this inequality we can show ultracontractivity estimates of the associated heat semigroup. This uses ideas going back to Nash \cite{Nash-58}. We refer to \cite{Davies-89} and references therein for further details and only sketch the proof.
\begin{Theorem}
There exists $c>0$ such that we have for all $u\in L^2 (\CX_\Gamma)$ and $t>0$
\[\|e^{-t \Delta } u\|_\infty \leq c t^{-\frac{1}{4}} \|u\|_2.\]
Moreover the heat semigroup $e^{-\Delta t}$ has an integral kernel $p(t,x,y)$ for all $t>0$ which satisfies
\[0\leq p(t,x,y) \leq c t^{-\frac{1}{2}}.\]
Furthermore this kernel is $\CC^\infty ((0,\infty)\times \CX_e \times \CX_{e'})$ for all $e,e' \in E$, all even derivatives w.r.t. space are continuous and all odd derivatives w.r.t. space satisfy the Kirchhoff boundary conditions.
\end{Theorem}
\begin{proof} The first statement follows from Theorem 2.4.6 \cite{Davies-89} and the estimate on the heat kernel from Theorem 2.1.2. \cite{Davies-89}. We only prove the last claim. We prove this along the lines of \cite{Davies-89} Theorem 5.2.1.\\
By the spectral theorem for all $f\in L^2(\CX_\Gamma)$ we have that $e^{-t\Delta} f\in \CD(\Delta^n)$ for all $n\geq 1$ and $t>0$. By Sobolev's inequality we have that for all $e\in E$ and $t>0$ the mapping $e^{-t\Delta}f \big|_e$ is $\CC^\infty$. Furthermore all even derivatives w.r.t. space are continuous and all odd derivatives w.r.t. space satisfy the Kirchhoff b.c. in each vertex. Since the mapping $t\mapsto e^{-t\Delta } f \in \CD(\Delta^n)$ is analytic, we can conclude that $(t,x)\mapsto e^{-t\Delta} f \big|_e$ is $\CC^\infty$ for all $e\in E$. For fixed $t>0$ and $x\in \CX_\Gamma$ the map $f\mapsto e^{-t\Delta } f(x)$ is bounded on $L^2$, hence there exists $a(t,x)\in L^2$ such that
\[e^{-t\Delta } f(x) = \langle f, a(t,x)\rangle.\]
This implies that $(t,x)\to a(t,x)$ is weakly infinitely differentiable on each edge interior and hence (\cite{Davies-89}) also in norm. This implies norm continuity w.r.t. space for all even derivatives and the Kirchhoff b.c. for all odd ones. For $g\in \CC_c^\infty$ we have by construction
\[\langle e^{-t\Delta} f, g\rangle = \int \langle f, a(t,x)\rangle g(x) \:dm (x)\]
and hence the $L^2$ identity
\[e^{-t \Delta } g = \int  a(t,x) g(x) \:dm(x).\]
Let now $g,h \in \CC_c^\infty$ then
\begin{eqnarray*}
\langle e^{-t\Delta} g,h \rangle &=& \langle e^{-\frac{t}{2}\Delta} g,e^{-\frac{t}{2}\Delta} h \rangle\\
&=& \int p(t,x,y) g(x)\bar{h}(y) \: dm(x) dm(y)
\end{eqnarray*}
with the identity
\[p(t,x,y) = \langle a(\tfrac{t}{2},x), a(\tfrac{t}{2}, y)\rangle\]
\end{proof}
Now applying Theorem 2.2.3 \cite{Davies-89} to the ultracontractivity estimate above, we arrive at a logarithmic Sobolev inequality, that is for all $0\leq u \in W^{1,2}_0(\CX_\Gamma)\cap L^1(\CX_\Gamma)$ we have
\[\int u^2 \log f \: dm(x) \leq \epsilon \CE(u) + M(\epsilon) \|f\|^2 + \|f\|_2^2\log \|f\|_2\]
for all $\epsilon >0$ and with $M(\epsilon)= c - \frac{\log \epsilon}{4}$. This implies Gaussian upper bounds for the heat kernel as shown in § 3.2 of \cite{Davies-89}, i.e. there exist constants  $a,b>0$ such that we have
\[0\leq p(t,x,y) \leq a t^{-\frac{1}{4}} e^{-b \frac{\rho(x,y)^2}{4t}}\]
where $\rho(x,y)$ is the intrinsic distance defined in section 1 and which coincides with the distance function $d(x,y)$ by Proposition \ref{intrinsic}.
\begin{remark}{\rm In the case of compact graphs, the method of \cite{Davies-89} was used in \cite{Mugnolo-06} to obtain similar results. However the method there to establish a Nash inequality does not seem to work in the non-compact case.\\
In \cite{KostrykinPS-07} explicit power series expansion for the heat kernel were obtained for graphs with finitely many vertices, i.e. for graphs where edges of infinite lengths are allowed.}
\end{remark}
\section{Volume doubling and Poincaré inequality}
In spite of one-dimensional behavior discussed in the previous section, the combinatorial structure causes some different local dimension effects in the sense of the following lemma. As we will see later, these result provide the necessary framework to use the results from \cite{Sturm-95} and \cite{Sturm-96}.
\begin{Lemma}[Local volume doubling]
Let $x\in \CX_\Gamma$ and fix $R>0$ such that $|B_{2R}(x)\cap \{v\in V: d_v >2\}| \leq 1$. Then for all $0<r<R$ we have
\[0< m(B_{2r}(x)) \leq (\tfrac{d_v}{2}+1) \cdot m(B_r(x)) < \infty.\]
In particular the volume doubling holds locally, that is, for all relative compact $Y\subset \CX_\Gamma$ there exists $\nu =\nu(Y)$ such that for all $x\in \CX_\Gamma$, $r>0$ with $B_{2r}(x) \subset Y$ we have
\[0< m(B_{2r}(x)) \leq 2^\nu \cdot m(B_r(x)) < \infty.\]
\end{Lemma}
\begin{proof}
Without restriction $|B_{2R}(x)\cap \{v\in V: d_v >2\}| =1$, since the other case is trivial. Denote the vertex contained in this intersection by $v$. Then we have
\begin{eqnarray*}
m(B_{2r}(x)) &=& m(B_{d(x,v)}(x)) + m(B_{2r-d(v,x)}(v))\\
&=& 2d(v,x) + d_v (2r-d(v,x))\\
&=& (2-d_v) d(v,x) + 2rd_v\\
&\leq& (2 - d_v) r + 2d_vr\\
&\leq& (1+\tfrac{d_v}{2})\cdot m(B_r(x)).
\end{eqnarray*}
The second claim follows easily from the first, since by relative compactness there exists a lower bound for the edge lengths and an upper bound for the vertex degrees.
\end{proof}
\begin{Theorem}[Local Poincaré inequality]\label{PI}
Let $x\in \CX_\Gamma$ and fix $R>0$ such that $|B_{R}(x)\cap \{v\in V: d_v >2\}| \leq 1$. Then for all $0<r<R$ and $u\in W^{1,2} (B_r(x))$ we have
\[ \int\limits_{B_r(x)} |u(y) - \bar{u}_{B_r(x)}|^2 \: dm(y) \leq c_P r^2 \int\limits_{B_r(x)} |u'(y)|^2\:dm(y)\]
with $c_P = d_v$ if $v\in B_R (x) \cap V$ or $c_P =2$ else.\\
In particular the Poincaré inequality holds locally, that is, for all relative compact $Y\subset \CX_\Gamma$ there exists $c_P =c_P (Y)$ such that for all $x\in \CX_\Gamma$, $r>0$ with $B_{r}(x) \subset Y$ we have
\[ \int\limits_{B_r(x)} |u(y) - \bar{u}_{B_r(x)}|^2 \: dm(y) \leq c_P r^2 \int\limits_{B_r(x)} |u'(y)|^2\:dm(y)\]
for all $u\in W^{1,2}(B_r(x))$.
\end{Theorem}
\begin{proof}
Both inequalities follow from Corollary \ref{1.2} with $p=q=2$, since $\diam B_r(x) =2r$ and in the first case we have $m(B_r(x))\leq d_v \cdot r$. In general we have
\[ \int\limits_{B_r(x)} |u(y) - \bar{u}_{B_r(x)}|^2 \: dm(y) \leq m(B_r(x)) r \int\limits_{B_r(x)} |u'(y)|^2\:dm(y).\]
But since $Y$ is relatively compact the quantity $c_P := \sup \{ \frac{m(B_r(x))}{r} \ : B_r(x) \subset Y\}$ is finite.
\end{proof}
If we assume a more uniform local structure, we can even prove stronger versions of the above statements. To do so, we need the following definition.
\begin{Definition}
Let $\CX_\Gamma$ a metric graph as above. We say $\CX_\Gamma$ is of bounded geometry if
\[ D:= \sup_{v\in V} d_v  < \infty \]
and
\[ \ell:= \inf_{e\in E} l(e) > 0.\]
\end{Definition}
As soon as the graph has a uniform lower bound on the edge lengths, we can also assume w.l.g. that the edge lengths are bounded uniformly from above by $2\ell$ by introducing vertices of degree $2$.\\
For graphs of bounded geometry the volume doubling property holds uniformly locally on $\CX_\Gamma$ for all $0<r<\frac{\ell}{4}$ with uniform doubling constant
\[c_D = 1+\tfrac{D}{2}.\]
Furthermore for fixed $L> \frac{\ell}{4}=: \ell'$  the doubling property then holds also uniformly for $0<r<L$ with a different constant depending on the choice of $L$. We briefly sketch the proof taken from \cite{SaloffCoste-02}: it suffices to show that for all $x\in\CX_\Gamma$ and all $L\geq \ell'$ we have
\[m(B_{{L}+\frac{\ell'}{2}}(x)) \leq c_D^2 m(B_{L}(x)) \label{VD}\tag{$\diamondsuit$}.\]
Consider therefore a maximal set $\CB$ of points in $B_{{L}-\frac{\ell'}{2}}(x)$ at distance at least $\frac{\ell'}{2}$ apart, which forces the balls $B_{\frac{\ell'}{2}}(y)$, $y\in  \CB$, to be disjoint and contained in $B_{L'}(x)$. Hence
\[\sum_{y\in\CB} m(B_\frac{\ell'}{2}(y)) \leq m(B_{L}(x)).\]
Since $\CB$ is maximal, the balls $B_{\ell'}(y)$ cover $B_{{L}-\frac{\ell'}{2}}(x)$ and in particular the balls $B_{2\ell'}(y)$ cover $B_{{L}+\frac{\ell'}{2}}(x)$. This yields
\[m(B_{{L}+\frac{\ell'}{2}}(x)) \leq \sum_{y\in \CB} m(B_{2\ell'}(y)).\]
Now by applying twice the doubling property we have \eqref{VD} with a doubling constant of at most $c_D^{\frac{4{L}}{\ell'}}$: choose $n\in \IN$ such that $(n-1)\frac{\ell'}{2} \leq L \leq n \frac{\ell'}{2}$. Then we have
\begin{eqnarray*}
m(B_L(x)) &\leq &c_D^2 m(B_{L-\frac{\ell'}{2}}(x)) \leq \dots\\
 &\leq& (c_D^2)^{n-1} m(B_{L- (n-1)\frac{\ell'}{2}}(x)) \leq c_D^{4\frac{L}{\ell'}} m(B_{\frac{\ell'}{2}}(x)).
\end{eqnarray*}
Summarizing the discussion above, we have proven the following corollary.
\begin{Corollary}[Uniform local volume doubling]
Assume the graph is of bounded geometry and fix $L>0$. Then there exists $\nu=\nu(\ell, D, L)$ such that for all  $x\in \CX_\Gamma$ and all $r\in (0,L)$  we have
\[0< m(B_{2r}(x)) \leq 2^\nu \cdot m(B_r(x)) < \infty.\]
\end{Corollary}
\begin{remark}{\rm
Associated with the doubling constant we define the local dimension to be $\nu = \frac{\log (D+2)}{\log 2} -1$. This is motivated by the fact that in $\IR^n$ we have ${\mathrm{vol}} (B_{2r}(x)) = 2^n {\mathrm{vol}}(B_r(x))$. So in this case the local dimension and the usual dimension agree. Note that every $\tilde{\nu}>\nu$ defines a local dimension as well and that the local dimension may change from point to point.\\
For further details concerning volume doubling and its consequences we refer to the book \cite{SaloffCoste-02}.}
\end{remark}
If the graph is of bounded geometry, then as consequence of Theorem \ref{PI}, the Poincaré inequality holds true uniform locally  as well:
\begin{Corollary}[Uniform local Poincaré inequality]
Assume the graph is of bounded geometry and fix $L>0$. Then there exists $c_P >0$ such that for all $0<r<L$ and $u\in W^{1,2} (B_r(x))$ we have
\[ \int\limits_{B_r(x)} |u(y) - \bar{u}_{B_r(x)}|^2 \: dm(y) \leq c_P r^2 \int\limits_{B_r(x)} |u'(y)|\:dm(y).\]
\end{Corollary}
\begin{remark}{\rm
Note that assuming the lower bound on the lengths forces the metric space $(\CX_\Gamma, d)$, and hence $(\CX_\Gamma, \rho)$, to be complete.}
\end{remark}
\section{The parabolic Harnack inequality}
The aim of this section is to state the Harnack inequality and to give a brief sketch of the ideas of its proof. These ideas go back to the fundamental work of Moser \cite{Moser-61}, \cite{Moser-63} and \cite{Moser-71} and were based on the validity of local Poincaré and Sobolev inequalities. However in \cite{SaloffCoste-92} it was shown that the doubling property and the Poincaré inequality suffice to show the local Sobolev inequality. In our setting, i.e. in the case of strongly local Dirichlet forms this was generalized in \cite{Sturm-95} and \cite{Sturm-96}. So as we have shown volume doubling and the Poincaré inequality in the previous sections, we obtain a Harnack inequality, see section 4.2 and for a brief sketch of its proof see section 4.3 as well.

\medskip

First we introduce the notion of weak solutions.
\subsection{Weak solutions}
Let $I\subset \IR$ be an open interval. We say that $u:I \to W^{1,2}(\CX_\Gamma)$, $t\mapsto u_t$ is in $L^2(I\to W^{1,2}(\CX_\Gamma))$ if $u$ is Bochner measurable and
\[ \int\limits_I \|u_t\|_{1,2}^2 \: dt < \infty.\]
Denote by $W^{-1,2} (\CX_\Gamma)$ the dual of $W^{1,2}(\CX_\Gamma)$ with respect to the inner product of $L^2(\CX_\Gamma)$, such that we have
\[W^{1,2}(\CX_\Gamma) \subset L^2(\CX_\Gamma) \subset W^{-1,2}(\CX_\Gamma).\]
Since the embedding $W^{1,2}(\CX_\Gamma) \hookrightarrow L^2(\CX_\Gamma)$ is dense and continuous, the same is true for the embedding of $L^2(\CX_\Gamma) \hookrightarrow W^{-1,2}(\CX_\Gamma)$. As usual we say that $v: I \to W^{-1,2}(\CX_\Gamma)$ is a generalized derivative of $u:I\to W^{1,2}(\CX_\Gamma)$ if for all $\phi\in C_c^\infty(I)$
\[\int\limits_I u(t) \phi'(t) \: dt = - \int\limits_I v(t) \phi(t)\: dt\]
and write $\tfrac{\partial}{\partial t} u$ for $v$. We set $\CW(I\times \CX_\Gamma):= \{u\in L^2(I\to \CX_\Gamma) : \ \exists g \in L^2(I\to W^{-1,2}(\CX_\Gamma))\mbox{ s.t. } u'=g\} $ with norm
\[\|u\|_{\CW} := \biggl( \int\limits_I \|u_t\|^2_{1,2} + \|\tfrac{\partial}{\partial t} u_t \|^2 _{-1,2}\: dt \biggr)^\frac{1}{2}.\]
We will deal with functions which are locally in $\CW(I\times \CX_\Gamma)$: Let $G$ be an open subset of $\CX_\Gamma$ and denote by $Q=I\times G$ the parabolic cylinder. Then $\CW_\loc (Q)$ consists of all $dm\otimes dt$-measurable functions on $Q$ such that for all relatively compact $G'\subset G$ and every relatively compact, open interval $I'\subset I$ there exists a function $\tilde u\in \CW(I\times \CX_\Gamma)$ with $u=\tilde u$ on $I'\times  G'$. We say that a function $u$ belongs to $\CW_o (I\times G)$ if $u\in\CW(I\times \CX_\Gamma)$ and if for a.e. $t\in I$ the function $u_t$ has compact support in $G$. Note that a function $u\in\CW_o(I\times G)$ has only to vanish on $I\times \partial G$, but neither on the upper nor on the lower boundary  $\partial I \times G$. Roughly speaking the space $\CW_\loc$ is the appropriate space of weak solutions and $\CW_o$ the appropriate space of test functions with respect to the heat equation:
\begin{Definition}\label{31}
We say that $u$ is a local subsolution (resp. local supersolution) of the equation
\begin{eqnarray*}
-\Delta u &=& \frac{\partial }{\partial t} u \quad\mbox{on $Q$}
\end{eqnarray*}
if $u\in \CW_\loc (Q)$ and
\[\tag{$*$}\label{1}
\CE_J(u,\phi) := \int\limits_J \CE(u,\phi) \: dt + \int\limits_J (\tfrac{\partial}{\partial t} u, \phi) \: dt \leq 0
\]
(or $\CE_J(u,\phi) \geq 0$, resp.) for all relatively compact, open $J\subset I$ and all nonnegative $\phi\in \CW_o(Q)$. The function $u$ is called a local solution if it is a local subsolution and a local supersolution.  In this case \eqref{1} holds true with ''$=$'' for all $\phi \in \CW_o(Q)$.
\end{Definition}
\begin{remark}{\rm We have seen in section 3 that in the case of a graph of bounded geometry we have
\[m(B_{L}(x))\leq c_D^{4\frac{L}{\ell'}} m(B_{\frac{\ell'}{2}}(x)).\]
for all $L>0$. Hence the volume growth is at most exponentially. This implies that
\[\int\limits_1^\infty \frac{r}{\ln m(B_r(x))} \: dr = \infty.\]
By Theorem 4 in \cite{Sturm-94} this implies that $\CE$ is stochastically complete, which is equivalent to each one of the following:
\begin{itemize}
\item $T_t 1 = 1$ for some (and then for all) $t>0$, where $T_t$ denotes the associated semigroup to $\CE$.
\item For some (and then for all) $\alpha >0$ every nonnegative solution $u \in L^\infty(\CX_\Gamma)$ of the equation $(\Delta + \alpha) u =0$ is identically 0.
\end{itemize}
For more details on stochastical completeness see \cite{Sturm-94,Grigoryan-99}.}
\end{remark}
\subsection{Harnack inequality}
\begin{Theorem}\label{HI}
Fix $0<\epsilon<\eta<\sigma<1$ and $\zeta \in (0,1)$. Assume the metric graph is of bounded geometry and fix $L>0$. Then there exists a constant $c_H$ such that for all balls $B_r(x)$, $0<r<L$, $x\in \CX_\Gamma$ and all $s\in \IR$ we have
\begin{equation}
\sup_{(t,y)\in Q_-} u(s,y) \leq c_H \inf_{(t,y)\in Q_+} u(s,y) \label{PHI} \tag{PHI}
\end{equation}
whenever $u$ is a nonnegative local solution of the equation $-\Delta u = \frac{\partial}{\partial t} u$ in $Q= (s, s+r^2) \times B_r(x)$, where $Q_- = (s+\epsilon r^2, s + \eta r^2 ) \times \zeta B_{\zeta r}(x)$ and $Q_+ = (s+\sigma r^2, s+ r^2) \times B_{\zeta r}(x)$. Here the constant $c_H$ depends only on $\nu$, $c_P$ and additionally on the choice of parameters $\epsilon$, $\eta$, $\sigma$ and $\zeta$, but not on $u$, $r$, $s$ and $x$.
\end{Theorem}
Of course, there are simple example of graphs which are not of bounded geometry, but then still a local version of the parabolic Harnack inequality is valid, see \cite{Sturm-96}.
\begin{Theorem}
Fix $0<\epsilon<\eta<\sigma<1$, $\zeta \in (0,1)$ and let $Y\subset \CX_\Gamma$ be precompact. Then there exists a constant $c_H = c_H(Y)$ such that for all $x\in \CX_\Gamma$, $r>0$ with $B_r(x)\subset Y$ and all $s\in \IR$ we have
\begin{equation*}
\sup_{(t,y)\in Q_-} u(t,y) \leq c_H \inf_{(t,y)\in Q_+} u(t,y)
\end{equation*}
whenever $u$ is a nonnegative local solution of the equation $-\Delta u = \frac{\partial}{\partial t} u$ in $Q= (s, s+r^2) \times B_r(x)$, where $Q_- = (s+\epsilon r^2, s + \eta r^2 ) \times B_{\zeta r}(x)$ and $Q_+ = (s+\sigma r^2, s+ r^2) \times  B_{\zeta r}(x)$. Here the constant $c_H$ depends only on $\nu(Y)$, $c_P(Y)$ and additionally on the choice of parameters $\epsilon$, $\eta$, $\sigma$ and $\zeta$, but not on $u$, $r$, $s$ and $x$.
\end{Theorem}
\subsection{The road to the parabolic Harnack inequality}
In this section we give a survey on the proof of the Harnack inequality \eqref{PHI} in the case of a graph of bounded geometry as in theorem \ref{HI}.

\medskip

{{\textbf{Assumption }} In what follows, we will assume that the graph is of bounded geometry, that is
\[ D:= \sup_{v\in V} d_v  < \infty \]
and
\[ \ell:= \inf_{e\in E} l(e) > 0.\]
We have seen in section 3, that this implies that the volume doubling and the Poincaré inequality hold locally uniform for all $0<r<L$, where $L>0$ is fixed.}

\medskip

As mentioned above, one crucial ingredient for the proof of (PHI) is a local version of the Sobolev inequality.
It was proven by Saloff-Coste, see \cite{SaloffCoste-92}, that the volume doubling and Poincaré inequality suffice to show it.
\begin{Theorem}
There exists a constant $c_S>0$ depending only on $\nu$ and $c_P$ such that we have
\[\biggl( \fint\limits_{B_r(x)} |u|^\frac{2\nu}{\nu-2} \: dm \biggr)^\frac{\nu-2}{\nu} \leq c_S   r^2 \biggl( \fint\limits_{B_r(x)} |u'(x)|^2\: dm(x) + r^{-2} \fint\limits_{B_r(x)} |u|^2 \:dm\biggr)\]
for all $0<r<L$, $x\in \CX_\Gamma$ and $u\in W^{1,2}(B_r(x))$ with $\supp u \subset B_r(x)$.
\end{Theorem}
Hence one has all tools to imitate the original proof by Moser to obtain the sub- and supersolution estimates. In the context of Riemannian manifolds this was done in \cite{SaloffCoste-95}. In the case of strongly local Dirichlet spaces it was done in \cite{Sturm-95} and yields in our setting:
\begin{Theorem}
\begin{itemize}
\item Fix $1<p<\infty$ and let $\delta \in (0,1)$. Then, for any nonnegative local subsolution $u$ of the equation $-\Delta u = \frac{\partial}{\partial t} u$ in $Q = (s, s+r^2) \times B_r (x)$ we have the estimate
\[ \sup_{Q_\delta} u^p \leq C \delta^{-\nu-2} (r^2 m(B_r(x)))^{-1} \iint\limits_Q |u|^p \: dm \: dt,\]
with $Q_\delta = (s+\delta r^2, s+r^2 )\times B_{(1-\delta)r}(x)$, where the constant $C>0$ is independent of $u$, $\delta$, $s$ and the ball of radius $0<r<L$.
\item Fix $0<p<\infty$ and let $\delta \in (0,1)$. Then for any nonnegative local supersolution $u$ of the equation $-\Delta u = \frac{\partial}{\partial t} u$ in $Q = (s, s+r^2) \times B_r(x)$ we have the estimate
\[ \sup_{Q_\delta} u^{-p} \leq C \delta^{-\nu-2} (r^2 m(B_r(x)))^{-1} \iint\limits_Q |u|^{-p} \: dm \: dt \]
with $Q_\delta = (s+\delta r^2, s+r^2 )\times B_{(1-\delta)r}(x)$, where the constant $C>0$ is independent of $u$, $\delta$, $s$ and of the ball of radius $0<r<L$.
\item Fix $0<p_0< 1+ \frac{2}{\nu}$. Then for all $0<\delta<1$ and all $0< p \leq p_0$ any nonnegative local supersolution $u$ of the equation $-\Delta u = \frac{\partial}{\partial t} u$ in $Q = (s, s+r^2) \times B_r(x)$ satisfies
\[ \Bigl(\iint\limits_{Q_{\delta}'} |u|^{p_0}\: dm dt\Bigr)^\frac{p}{p_0} \leq \Bigl( C \delta ^{-\nu -2} (r^2 m(B_r(x)))^{-1}\Bigr)^{1-\frac{p}{p_0}} \iint\limits_Q |u|^p \: dm dt,\]
where $Q_\delta ' = (s, s+ (1-\delta)r^2) \times B_{(1-\delta)r}(x)$. Here, the constant $C$ is independent of $\delta$, $p$, $u$, $s$ and of the ball of Radius $0<r<L$ but it depends on $p_0$.
\end{itemize}
\end{Theorem}
These results can be viewed as the easy part of the proof. The more technical one is to deduce from the supersolution estimates a weak Harnack inequality. A basic tool will be the following weighted Poincaré inequality, which is once again a consequence of the volume doubling and the strong Poincaré inequality, see \cite{SaloffCoste-95} and \cite{Sturm-96}.
\begin{Lemma}
Fix $\delta \in (0,1)$. Then there exists a constant $c_{P'}$ such that for all $0<r< L$, $x_0\in \CX_\Gamma$ and $u\in W^{1,2} (B_r(x))$
\[\int\limits_{B_r(x_0)} |u-u_{B,\psi}| \psi^2 \: dm \leq C_P' r^2 \int\limits_{B_r(x)} \psi^2 d\Gamma(u)\]
where $\psi (y)= (1-\frac{d(y,B_{(1-\delta)r }(x))}{\delta r} )_+$ and $u_{B,\psi} = \int\limits_{B_r(x)} u \psi^2 \: dm / \int\limits_{B_r(x)} \psi^2 \: dm$. Here the constant $C_p '$ depends only on $C_p$, $\nu$ and $\delta$.
\end{Lemma}
Using this inequality, one obtains estimates of the size of the logarithm of a local solution. We denote by $m\otimes \lambda$ the product measure on $\CX_\Gamma \times \IR$.
\begin{Lemma}
Fix $\delta,\tau \in (0,1)$. Then for any nonnegative local supersolution $u$ of the equation $-\Delta u = \frac{\partial}{\partial t} u$ in $Q = (s, s+r^2) \times B_r(x)$, there exists a constant $c=c(u,\tau)$ such that, for all $\mu >0$,
\[ (m\otimes \lambda)( \{ (t,z) \in K_+  :  \log u < -\mu -c \}) \leq C r^2 m(B_r(x)) \mu^{-1}\]
and
\[ (m\otimes \lambda)( \{ (t,z) \in K_-  :  \log u > \mu -c \}) \leq C r^2 m(B_r(x)) \mu^{-1}\]
where $K_+ = (s+\tau r^2, s+r^2) \times B_{(1-\delta)r}(x)$ and $K_- = (s, s+ \tau r^2)\times B_{(1-\delta)r}(x)$. Here the constant $C$ is independent of $\mu >0$, $u$, $s$ and of the ball of radius $0<r<L$.
\end{Lemma}
Combining this estimate on $\log u$ with the supersolution estimates, one obtains via the so-called abstract lemma, see \cite{SaloffCoste-02}, the following weak Harnack inequality.
\begin{Theorem}
Let $0< p < 1 + \frac{2}{\nu}$, $0<\epsilon<\eta<\sigma<1$ and $0<\zeta<1$. Then any positive local supersolution $u$ of $-\Delta u = \frac{\partial}{\partial t} u$ in $Q=(s, s+r^2) \times B_r(x)$, satisfies
\[ \Bigl(\iint\limits_{Q_-} |u|^p \: dmdt \Bigr)^\frac{1}{p} \leq C (r^2m(B_r(x)))^\frac {1}{p} \inf_{Q_+} u,\]
where
\[Q_- = (s+\epsilon r^2, s+ \eta r^2)\times B(x,\zeta r) \mbox { and } Q_+ = (s+\sigma r^2, s+ r^2) \times B(x,\zeta r).\]
Here the constant $C$ is independent of $u$, $s$ and of the ball $B_r(x)$ of radius $0<r<L$.
\end{Theorem}
As mentioned already, this together with the subsolution estimate finally yields the Harnack inequality.
\subsection{Consequences}
In this subsection we collect some well known consequences of the Harnack inequality to underline its power. Throughout we assume that the volume doubling and the Poincaré inequality are satisfied local uniform. In the case that both properties are just satisfied locally, one has to change the assumptions in an obvious way. Note that in that case all intrinsic constants depend heavily on the subset $Y$. For the proofs we refer to chapter 5 of \cite{SaloffCoste-02}.\\
We start with the following Hölder estimates for local solutions of the heat equation.
\begin{Proposition}\label{Hoelder}
There exists constants $\alpha \in (0,1)$ and $C>0$ such that for all $0<2r<L$ and $T\in \IR$ we have
\[|u(s,y)-u(t,x)| \leq C \sup_{Q} |u| \cdot \left( \frac{|s-t|^\frac{1}{2} + d(y,z)}{r}\right)^\alpha\]
for all local solutions $u$ of $-\Delta u = \frac{\partial}{\partial t} u$ on $Q = (T-4r^2, T ) \times B_{2r}(x)$ and $(s,y), (t,z)\in (T-r^2,T)\times B_r(x)$.
\end{Proposition}
If one considers time independent solutions one directly deduces the elliptic Harnack inequality:
\begin{Proposition}
There exists a constant $C>0$ such that for all $0<2r<L$ we have
\[\sup_{B_r(x)} u \leq C \inf_{B_r(x)} u\]
for all nonnegative local solutions $u$ of $\Delta u = 0$ on $B_{2r}(x)$. Furthermore if $Y\subset\CX_\Gamma$ is relatively compact, then there exists a constant $C_Y>0$ such that
\[\sup_{Y} u \leq C_Y \inf_{Y} u\]
for all nonnegative local solutions $u$ of $\Delta u= 0$ on $Y$.
\end{Proposition}
\begin{proof}
The first statement follows easily from the parabolic Harnack inequality. The second statement follows by a chain of balls argument.\end{proof}
\begin{remark}{\rm
For a different proof of the elliptic Harnack inequality see also \cite{Haeseler-08}, which uses the Moser iteration technique developed in \cite{Moser-61}.}
\end{remark}
Clearly, by Proposition \ref{Hoelder}, also Hölder estimates for solutions of the time independent equation hold true.
\begin{Proposition}
There exists constants $\alpha \in (0,1)$ and $C>0$ such that for all $0<2r<L$ we have
\[|u(y)-u(z)| \leq C \sup_{B_{2r}(x)} |u| \left(\frac{|y-z|}{r}\right)^\alpha\]
for all nonnegative local solutions $u$ of $\Delta u = 0$ on $B_{2r}(x)$ and $y,z\in B_r(x)$.
\end{Proposition}
\begin{Proposition}
Assume now that the volume doubling and the Poincaré inequality hold uniformly in $r>0$ on $\CX_\Gamma$. Then $\CX_\Gamma$ has the strong Liouville property, that is any solution of $\Delta u=0$ on $\CX_\Gamma$ that is bounded from below (or above) is constant. Moreover, there exists an $\alpha >0$ such that any solution $u$ of $\Delta u =0$ on $\CX_\Gamma$ with the property
\[\lim_{r\to \infty} \frac{1}{r^\alpha} \sup_{B_r(x)} |u|=0\]
for some $x\in \CX_\Gamma$, must be constant.
\end{Proposition}
The next consequence of Harnack's inequality improves the upper Gaussian bound of the heat kernel from section 2.2. Moreover, together with the stated lower bound, this estimates is in fact equivalent to the local uniform parabolic Harnack inequality and hence to the local uniform volume doubling and the Poincaré inequality. For the proof we refer to \cite{Sturm-95,Sturm-96} (see \cite{SaloffCoste-95} as well).
\begin{Proposition}
There exists constants $c_1,c_2,C_1,C_2 >0$ such that for all $x,y\in \CX_\Gamma$ and all $0<t<L^2$ we have
\[\frac{c_1}{m(B_{\sqrt{t}}(x))} \exp {\left(- \frac{d(x,y)^2}{C_1t}\right)} \leq p(t,x,y) \leq \frac{c_2}{m(B_{\sqrt{t}}(x))} \exp{\left(- \frac{d(x,y)^2}{C_2t}\right)}.\]
\end{Proposition}
\section{Examples}
In this section, we have a look at examples which appear frequently in the literature.
\subsection{Generalized star graphs}
We call a graph $\CX_\Gamma$ a generalized star graph if we can decompose the edge set
\[ E = E_f \cup \bigcup E_i\]
where the set $E_f$ is finite set of edges of finite length and the sets $E_i$ give a finite set of the half-lines, see fig. 1 for a star center $E_f$. Hence the graph is of linear volume growth and the quantity $\sup\limits_{x,r} \tfrac{m(B_r(x))}{r}$ appearing in Poincaré's inequality exists. Moreover, the Harnack inequality holds uniform in $r>0$ and $x\in \CX_\Gamma$. These graphs where studied in \cite{KostrykinS-99}, \cite{KostrykinS-06} and \cite{KostrykinPS-07}.
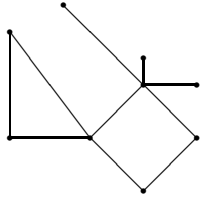
\begin{figure}
\setlength{\unitlength}{1pt}
\begin{center}
\begin{picture}(90,63)(-10,-17)
\put(0,0){\circle*{2}}
\put(0,40){\circle*{2}}
\put(30,0){\circle*{2}}
\put(70,0){\circle*{2}}
\put(50,20){\circle*{2}}
\put(50,-20){\circle*{2}}
\put(20,50){\circle*{2}}
\put(50,30){\circle*{2}}
\put(70,20){\circle*{2}}
\put(0,0){\line(1,0){30}}
\put(0,0){\line(0,1){40}}
\put(0,40){\line(3,-4){30}}
\put(30,0){\line(1,1){20}}
\put(30,0){\line(1,-1){20}}
\put(50,20){\line(1,-1){20}}
\put(50,-20){\line(1,1){20}}
\put(50,20){\line(-1,1){30}}
\put(50,20){\line(0,1){10}}
\put(50,20){\line(1,0){20}}
\end{picture}
\end{center}
\caption{A compact graph}
\end{figure}
\subsection{The Euclidean lattice $\IZ^d$}
As a second family of examples we consider the lattices $\IZ^d$, $d\in \IN$. That is $V = \IZ^d$ and $n\sim m$ if $|n-m|_1=1$, where $|\cdot|_1$ denotes the $\ell_1$-norm in $\IR^d$. Moreover each edge has length one, see fig. 2. An easy calculation shows, that the volume doubling holds uniform. Rather obvious is, that the Poincaré inequality holds uniform as well. This was shown in \cite{PivarskiSaloffCoste-08} in a more general context of Cayley graphs associated with finitely generated groups of polynomial volume growth.\\
Consider now the metric graph $\tilde\IZ^d$, which has the same vertex and edge set. We then assign to every edge a length $l\in [l_-, l_+]$ with $l_-,l_+ >0$. We then define a quasi-isometry in a canonical way: denote by $\pi_e$ the homeomorphism mapping $\CX_e\subset \IZ^d$ to $[0,1]$, and denote by $\tilde{\pi}_e$ the homeomorphism mapping $\tilde{\CX}_e\subset \tilde{\IZ}^d$ to $[0,l_e]$. Let $L_e:[0,1]\to [0,l_e],\ x\mapsto l_e x$. The quasi-isometry $J$ is then defined as $J_e= \pi_e^{-1} \circ L_e^{-1} \circ \tilde{\pi}_e$. Hence $l_- d(x,y) \leq \tilde{d}(Jx,Jy) \leq l_+ d(x,y)$ and therefore we have $B_{l_- r}(x) \subset J(B_r(x)) \subset B_{l_+ r}(x)$. This implies that the volume doubling holds globally on $\tilde{\IZ}^d$, where the doubling constant depends now additionally on $l_-,l_+$. By the chain rule we see that the Poincaré inequality also holds uniform. Hence the Harnack inequality holds uniform and the asymptotic heat kernel behavior is the same as in $\IZ^d$.\\
\begin{figure}
\setlength{\unitlength}{1pt}
\begin{center}
\begin{picture}(60,60)(-12,-8)
\multiput(0,0)(0,10){5}{ \multiput(0, 0)(10, 0){5}{
  \put(0,0){\line(0, 1){10}}
  \put(0,0){\line(1, 0){10}}
  \put(0,0){\circle*{2}}
  \put(0,-1){\line(0, -1){10}}
  \put(-1,0){\line(-1, 0){10}}}}
\end{picture}
\end{center}
\caption{The metric graph $\IZ^2$}
\end{figure}
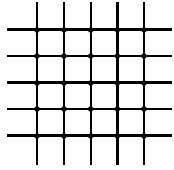

\subsection{Trees}
Typical examples of metric graphs which do not satisfy volume doubling or Poincaré inequality globally are trees. We call a metric graph a tree, if it contains no cycles, i.e. there exists no continuous mapping $\gamma:[0,1]\to \CX_\Gamma$ with $\gamma(0)=\gamma(1)$ and which is injective on $(0,1)$ (see Fig. 3 for a finite tree). Assume that the vertex degree is at least $2$ for all vertices $v\in V$ and that the edge lengths are equal to $1$. Then the volume is of exponentially growth. Moreover, the doubling property cannot hold uniformly. If we assume further that the vertex degrees are uniformly bounded from above, then we know by section 3 that the doubling property and the Poincaré inequality, and hence Harnack's inequality, hold locally uniform.
\begin{figure}
\begin{center}
\setlength{\unitlength}{1pt}
\begin{picture}(40,80)(10,-40)
\put(0,0){\line(1, 1){20}}     \put(0,0){\line(1, -1){20}}      \put(0,0){\circle*{2}}
\put(20,20){\line(2, 1){20}}   \put(20,20){\line(2, -1){20}}    \put(20,20){\circle*{2}}
\put(20,-20){\line(2, 1){20}}  \put(20,-20){\line(2, -1){20}}   \put(20,-20){\circle*{2}}
\put(40,30){\line(4, 1){20}}   \put(40,30){\line(4, -1){20}}    \put(40,30){\circle*{2}}
\put(40,10){\line(4, 1){20}}   \put(40,10){\line(4, -1){20}}    \put(40,10){\circle*{2}}
\put(40,-30){\line(4, 1){20}}  \put(40,-30){\line(4, -1){20}}   \put(40,-30){\circle*{2}}
\put(40,-10){\line(4, 1){20}}  \put(40,-10){\line(4, -1){20}}   \put(40,-10){\circle*{2}}
\end{picture}
\end{center}
\caption{A tree}
\end{figure}

\bigskip

\textbf{Acknowledgements.} The author takes the opportunity to thank Daniel Lenz for inspiring discussions, helpful suggestions and his guidance during this work. Moreover he would like to thank Peter Stollmann for introducing him to the topic which formed the basis on which this work is built.

\end{document}